\documentclass[letterpaper, 10pt, conference]{ieeeconf}

\IEEEoverridecommandlockouts
\usepackage{amsmath,amssymb,pstricks,color,dsfont}
\usepackage{graphicx}
\usepackage{psfrag,graphicx,epsfig}
\usepackage{xspace}
\usepackage{subfig}
\usepackage{float}
\usepackage{placeins}
\usepackage{multirow}
\usepackage{pgf,tikz}

\usepackage[vlined,linesnumbered,boxruled]{algorithm2e}

\usepackage[bookmarks=true]{hyperref}
\hypersetup
{
colorlinks=true,
linkcolor=red, citecolor=blue, filecolor=magenta, urlcolor=blue
}
\usepackage{url, cite}

\usepackage{array}

\newcolumntype{x}[1]{%
>{\centering\hspace{0pt}}p{#1}}%

\newcommand{\EQ}{\begin{eqnarray}}
\newcommand{\EN}{\end{eqnarray}}
\newcommand{\EQQ}{\begin{eqnarray*}}
\newcommand{\ENN}{\end{eqnarray*}}

\newcommand{\bremark}{\begin{remark} \begin{rm} }
\newcommand{\eremark}{ \end{rm} \rule{1mm}{2mm}
\end{remark} }
\newcommand{\btheorem}{\begin{theorem} \begin{rm} }
\newcommand{\etheorem}{ \end{rm} \rule{1mm}{2mm}
\end{theorem} }
\newcommand{\blemma}{\begin{lemma} \begin{rm} }
\newcommand{\elemma}{ \end{rm} \rule{1mm}{2mm}
\end{lemma} }
\newcommand{\bcorollary}{\begin{corollary} \begin{rm} }
\newcommand{\ecorollary}{ \end{rm} \rule{1mm}{2mm}
\end{corollary} }
\newcommand{\bdefinition}{\begin{definition}\begin{rm} }
\newcommand{\edefinition}{ \end{rm} \rule{1mm}{2mm}
\end{definition} }
\newcommand{\bproposition}{\begin{proposition} \begin{rm} }
\newcommand{\eproposition}{ \end{rm} \rule{1mm}{2mm}
\end{proposition} }
\newcommand{\bexample}{\begin{example} \begin{rm} }
\newcommand{\eexample}{ \end{rm} \rule{1mm}{2mm}
\end{example} }
\newcommand{\basm}{\begin{assumption} \begin{rm}}
\newcommand{\easm}{\end{rm} 
\end{assumption}}

\newcommand{\real}{\mathds{R}}

\newtheorem{theorem}{\bf Theorem}[section]
\newtheorem{lemma}{\bf Lemma}[section]
\newtheorem{definition}{\bf Definition}[section]
\newtheorem{remark}{\bf Remark}[section]
\newtheorem{corollary}{\bf Corollary}[section]
\newtheorem{proposition}{\bf Proposition}[section]
\newtheorem{example}{\bf Example}[section]
\newtheorem{assumption}{\bf Assumption}[section]

\newcommand\oprocendsymbol{\hbox{$\bullet$}}
\newcommand\oprocend{\relax\ifmmode\else\unskip\hfill\fi\oprocendsymbol}

\definecolor{darkgreen}{rgb}{0,0.6,0}
\definecolor{brown}{rgb}{0.65,0.16,0.16}

\date{}

\providecommand{\IncMargin}[1]{}

\begin{document}

\title{Game theoretic controller synthesis for multi-robot motion planning\\
Part I : Trajectory based algorithms}

\author{Minghui Zhu, Michael Otte, Pratik Chaudhari, Emilio Frazzoli \thanks{M. Zhu is with the Department of Electrical Engineering, Pennsylvania State University, 201 Old Main, University Park, PA, 16802.\newline Email:~\href{mailto:muz16@psu.edu}{muz16@psu.edu}.\newline M. Otte, P. Chaudhari and E. Frazzoli are with the Laboratory for Information and Decision Systems, Massachusetts Institute of Technology, 77 Massachusetts Avenue, Cambridge MA, 02139.\newline Email:~\mbox{\href{mailto:ottemw@mit.edu}{ottemw@mit.edu}}, \href{mailto:pratikac@mit.edu}{pratikac@mit.edu}, \href{mailto:frazzoli@mit.edu}{frazzoli@mit.edu}.\newline This research was supported in part by ONR Grant \#N00014-09-1-0751 and the Michigan/AFRL Collaborative Center on Control Sciences, AFOSR grant \#FA 8650-07-2-3744.}}

\maketitle

\begin{abstract} We consider a class of multi-robot motion planning problems where each robot is associated with multiple objectives and decoupled task specifications. The problems are formulated as an open-loop non-cooperative differential game. A distributed anytime algorithm is proposed to compute a Nash equilibrium of the game. The following properties are proven: (i) the algorithm asymptotically converges to the set of Nash equilibrium; (ii) for scalar cost functionals, the price of stability equals one; (iii) for the worst case, the computational complexity and communication cost are linear in the robot number.
\end{abstract}

\section{Introduction}

Robotic motion planning is a fundamental problem where a control sequence is found to steer a robot from the initial state to the goal set, while enforcing the environmental rules. It is well-known that the problem is computationally challenging~\cite{Reif:79}. The situation is even worse for multi-robot motion planning since the computational complexity exponentially grows as the robot number.


For multi-robot motion planning, non-cooperative game theoretic controller synthesis is interesting in two aspects: \emph{descriptive} and \emph{perspective}. From the descriptive point of view, Nash equilibrium is desirable in inherently competitive scenarios. More specifically, Nash equilibrium characterizes the stable scenarios among inherently self-interested players where none can benefit from unilateral deviations. From the perspective point of view, non-cooperative game theoretic learning holds the promise of providing computationally efficient algorithms for multi-robot controllers where the robots are assumed to be self-interested. Although Nash equilibrium may not be socially optimal, game theoretic approaches remain useful when the computational efficiency dominates.


There have been limited results on rigorous analysis of game theoretic controller synthesis for multi-robot motion planning. The paper~\cite{LaValle.Hutchinson:98} tackles multi-robot motion planning in the framework of feedback differential games. However, it lacks of the rigorous analysis of the algorithm's convergence and computational complexity. In addition, static game theory has been used to synthesize distributed control schemes to steer multiple vehicles to stationary and meaningful configurations; e.g., in~\cite{Zhu.Martinez:12} for optimal sensor deployment, in~\cite{Arsie.Savla.ea:TAC09} for vehicle routing and in~\cite{GA-JRM-JSS:07a} for target assignment.



\emph{Contributions.} This paper presents the first distributed, anytime algorithm to compute open-loop Nash equilibrium for non-cooperative robots. More specifically, we consider a class of multi-robot motion planning problems where each robot is associated with multiple objectives and decoupled task specifications. The problems are formulated as an open-loop non-cooperative differential game. By leveraging the RRG algorithm in~\cite{SK-EF:11}, iterative better response and model checking, a distributed anytime computation algorithm, namely the iNash-trajectory algorithm, is proposed to find a Nash equilibrium of the game. We formally analyze the algorithm convergence, the price of stability as well as the computational complexity and communication cost. The algorithm performance is demonstrated by a number of numerical simulations. 




\emph{Literature review.} Sampling based algorithms have been demonstrated to be efficient in addressing robotic motion planning in high-dimension spaces. The Rapidly-exploring Random Tree (RRT) algorithm and its variants; e.g., in~\cite{SML-JJK:00,SML-JJK:01}, are able to find a feasible path quickly. Recently, two novel algorithms, PRM$^*$ and
RRT$^*$, have been developed in~\cite{SK-EF:11}, and shown to be computationally efficient and asymptotically optimal. In~\cite{Karaman.Frazzoli:12}, a class of sampling-based algorithms is proposed to compute the optimal trajectory satisfying the given task specifications in the form of deterministic $\mu$-calculus.


Regarding the multi-robot open-loop motion planning, the approaches mainly fall into three categories: centralized planning in; e.g.,~\cite{Sanchez.Latombe:02,Xidias.Aspragathos:08}, decoupled planning in; e.g.,~\cite{Kant.Zucker:86,Simeon.Leroy.Lauumond:02} and priority planning in; e.g.,~\cite{Buckley:89,Erdmann.Lozano:87}. Centralized planning is complete but computationally expensive. In contrast, decoupled and priority planning can generate solutions quicker, but are incomplete. However, the existing algorithms assume the robots are cooperative and are not directly applicable to compute Nash equilibrium where none of self-interested robots is willing to unilaterally deviate from.


Another set of relevant papers are concerned with numerical methods for feedback differential games. There have been a very limited number of feedback differential games whose closed-form solutions are known, including homicidal-chauffeur and the lady-in-the-lake games ~\cite{TB-GO:99,Isaacs:99}. The methods based on partial differential
equations; e.g., in~\cite{MB-ICD:97,MB-MF-PS:99,PES:99}, viability theory; e.g.,
in~\cite{JPA:09,JPA-AB-PSP:11,PC-MQ-PSP:99} and level-set methods, e.g., in~\cite{Sethian:96} have been proposed to determine numerical solutions to differential games. However, the papers aforementioned only study one-player and two-player differential games. Multi-player linear-quadratic differential games have been studied in; e.g.,~\cite{TB-GO:99}.




\section{Problem formulation}\label{sec:formulation}

Consider a team of robots, labeled by $\mathcal{V}_{R} \triangleq \{1,\cdots,N\}$. Each robot is associated with a dynamic system governed by the following differential equation: \begin{align}\dot{x}_i(t) = f_i(x_i(t), u_i(t)), \label{e1}\end{align} where $x_i(t)\in
\mathcal{X}_i\subseteq\real^{n_i}$ is the state of robot~$i$, and $u_i(t)\in\mathcal{U}_i$ is the control of robot~$i$. For system~\eqref{e1}, the set of admissible control strategies for robot~$i$ is given by:
\begin{align*}\mathcal{U}_i \triangleq \{u_i(\cdot) \; : \; [0,+\infty)\rightarrow U_i,\;{\rm measurable}\},\end{align*} where $U_i\subseteq\real^{m_i}$.
For each robot~$i$, $\sigma^{[i]} : [0,+\infty) \rightarrow \mathcal{X}_i$ is a dynamically feasible trajectory if there are $T\geq0$ and $u_i : [0,T] \rightarrow \mathcal{U}_i$ such that: $(i)$ $\dot{\sigma}^{[i]}(t) = f_i(\sigma^{[i]}(t), u_i(t))$; $(ii)$ $\sigma^{[i]}(0) = x^{[i]}_{\rm init}$; $(iii)$ $\sigma^{[i]}(t)\in \mathcal{X}_i^F$ for $t\in[0,T]$; $(iv)$ $\sigma^{[i]}(T)\in\mathcal{X}_i^G$. The set of dynamically feasible trajectories for robot~$i$ is denoted by $\Sigma_i$. Note that the trajectories in $\Sigma_i$ do not account for the inter-robot collisions.

Let $\Pi_i$ be a finite set of atomic propositions and $\lambda_i : \mathcal{X}_i\rightarrow 2^{\Pi_i}$ associates each state in $\mathcal{X}_i$ with a set of atomic propositions in $\Pi_i$. Given a trajectory $\sigma^{[i]}$ of~\eqref{e1}, define by $T(\sigma^{[i]})$ the set of time instances when $\lambda_i(\sigma^{[i]}(t))$ changes. The word $w(\sigma^{[i]}) \triangleq \{w_0, w_1, \cdots\}$ generated by the trajectory $\sigma^{[i]}$ is such that $w_i = \lambda_i(\sigma^{[i]}(t_i))$ where $t_0 = 0$ and $T(\sigma^{[i]}) = \{t_1, t_2, \cdots\}$.

In this paper, we consider reachability tasks where each robot has to reach an open goal set $\mathcal{X}_i^G \subset \mathcal{X}$ and simultaneously maintain the state $x_i(t)$ inside a closed constraint set
$\mathcal{X}_i^F \subseteq \mathcal{X}$. As an example, let $\Pi_i = \{ p_G, p_F \}$ be the set of atomic propositions. The proposition $p_G$ is true if $x_i \in \mathcal{X}_i^G$ and similarly, $p_F$ is true if $x_i \in \mathcal{X}_i^F$. Consider an example task specification $\Phi_i$ expressed using the finite fragment of Linear Temporal Logic, (FLTL)~\cite{manna1995temporal} as $\Phi_i = \mathbf{F}\ p_G \wedge \mathbf{G}\ p_F$ where $\mathbf{F}$ is the eventually operator and $\mathbf{G}$ stands for the always operator. If the word formed by a trajectory $\sigma^{[i]}$ is such that for $w(\sigma^{[i]}) = w_0, w_1, \ldots, $ there exists some $w_k $ such that $p_G \in w_k$ and $p_F \in w_i$ for all $i \geq 0$, we say that the word $w(\sigma^{[i]})$ satisfies the LTL formula $\Phi_i$. Let us note that FLTL formulae such as those considered here can be automatically translated into automata. The word $w(\sigma^{[i]})$ satisfies the formula if it belongs to the language of the corresponding automaton. Please refer~\cite{baier2008principles} for a more thorough exposition of these concepts.
Denote by $[\Phi_i]\subseteq \Sigma_i$ the set of trajectories fulfilling $\Phi_i$. Each robot then determines a trajectory belonging to $[\Phi_i]$.

In addition to finding a trajectory that satisfies these specifications, the robot may have several other objectives such as reaching in the goal region in the shortest possible time.
To quantify these objectives, we define ${\tt Cost} : \Sigma \rightarrow \real_{\geq0}^p$ as the cost functional which maps each trajectory in $\Sigma\triangleq\bigotimes_{i\in \mathcal{V}_R}\Sigma_i$\footnote{$\bigotimes$ represents the product.} to a non-negative cost vector and each component of ${\tt Cost}$ corresponds to an objective of the robots. In what follows, we assume that ${\tt Cost}$ is continuous. In addition, the robots want to avoid the inter-robot collisions; i.e., keeping the state $x(t)$ outside the collision set $\mathcal{X}_{\rm col}$.

The above multi-robot motion planning problem is formulated as an open-loop non-cooperative game where each robot seeks to find a trajectory which is collision-free, fulfills its task specifications and minimizes the induced cost given the trajectories of other robots. That is, given $\sigma^{[-i]}\in \Sigma_{-i}$\footnote{We use $-i$ to denote all the robots other than $i$.}, each robot~$i$ wants to find a best trajectory in the feasible set ${\tt Feasible}_i(\Sigma_i,\sigma^{[-i]})\triangleq \{\sigma^{[i]}\in\Sigma_i\;|\; \sigma^{[i]}\in[\Phi_i],\;\; {\tt CollisionFreePath}(\sigma^{[i]},\sigma^{[-i]})=1\}$ where the procedure ${\tt CollisionFreePath}$ will be defined later. The solution notion we will use is Nash equilibrium formally stated as follows:

\begin{definition}[Nash equilibrium] The collection of trajectories $(\bar{\sigma}^{[i]})_{i\in \mathcal{V}_R}\in\Sigma$ is a Nash equilibrium if for any $i\in \mathcal{V}_R$, it holds that $\bar{\sigma}^{[i]}\in {\tt Feasible}_i(\Sigma_i,\bar{\sigma}^{[-i]})$ and there is no $\sigma^{[i]}\in {\tt Feasible}_i(\Sigma_i,\bar{\sigma}^{[-i]})$ such that ${\tt Cost}(\sigma^{[i]})\prec{\tt Cost}(\bar{\sigma}^{[i]})$\footnote{The relation $\preceq$ is defined on $\real^p$ and given by: for $a,b\in\real^p$, $a\preceq b$ if and only if $a_{\ell} \leq b_{\ell}$ for all $\ell\in\{1,\cdots,p\}$. Note that $\preceq$ is a partial order on $\real^p$.}.\label{def2}
\end{definition}

Intuitively, none of the robots can decrease its cost by \emph{unilaterally} deviating from a Nash equilibrium. Denote by $\Pi_{\rm NE}\subseteq \Sigma$ the set of Nash equilibria. Note that Definition~\ref{def2} is an extension of the standard one; e.g., in~\cite{TB-GO:99} where the cost functional of each player is scalar. We will compare Nash equilibrium with social (Pareto) optimum defined as follows:

\begin{definition}[Social (Pareto) optimum] The collection of trajectories $(\bar{\sigma}^{[i]})_{i\in \mathcal{V}_R}\in\Sigma$ is socially (Pareto) optimal if there is no $(\sigma^{[i]})_{i\in \mathcal{V}_R}\in \Sigma$ such that $\bigoplus_{i\in \mathcal{V}_R}{\tt Cost}(\sigma^{[i]}) \prec \bigoplus_{i\in \mathcal{V}_R}{\tt Cost}(\bar{\sigma}^{[i]})$.\footnote{$\bigoplus$ represents the summation.}\label{def1}
\end{definition}

Denote by $\Pi_{\rm SO}\subseteq \Sigma$ the set of social optimum. Throughout this paper, we assume that $\Pi_{\rm SO}\neq\emptyset$. In general, a Nash equilibrium may not be socially optimal. When ${\tt Cost}$ is scalar, the gap between the set of Nash equilibrium and the set of social optimum is usually characterized by price of anarchy and price of stability in; e.g.,~\cite{Nisan.Roughgarden.ea:07}.



\subsection{Primitives}

%

Here we define a set of primitives which will be used in the subsequent sections.


\paragraph{Sampling} The ${\tt Sample}(A)$ procedure returns uniformly random samples from set $A$.

\paragraph{Local steering} Given two states $x, y$, the ${\tt Steer}$ procedure returns a state $z$ by steering $x$ towards $y$ for at most $\eta > 0$ distance; i.e., ${\tt Stear}(x,y) \triangleq {\rm argmin}_{z\in \mathbb{B}(x,\eta)}\|z-y\|$.
In addition to this, we require that $\sigma(x, y)$, the trajectory connecting states $x$ and $y$, is such that $|T(\sigma(x,y))| \leq 1$, i.e., the label $\lambda(\sigma(x,y))$ changes at most once. This property of the local steering function is called \emph{trace inclusivity}~\cite{castro2013incremental}.



\paragraph{Nearest neighbor} Given a state $x$ and a finite set $S$ of states, the ${\tt Nearest}$ procedure returns the state in $S$ that is closest to $x$; i.e., $\displaystyle{{\tt Nearest}(S,x) \triangleq {\rm argmin}_{y \in S}\|y-x\|}$.

\paragraph{Near vertices} Given a state $x$, a finite set $S$ and a positive real number $r$, the ${\tt NearVertices}$ procedure returns the states in $S$ where each of them is $r$-close to $x$; ${\tt NearVertices}(S, x ,r) \triangleq \{y\in S\;|\; \|x-y\|\leq r\}$.

\paragraph{Path generation} Given a directed graph $G$ with a single root and no directed cycles, the ${\tt PathGeneration}(G)$ procedure returns the set of paths from the root to the leaf vertices.

\paragraph{Collision check of paths} Given a path $\sigma$ and a set of paths $\Pi$, the ${\tt CollisionFreePath}(\sigma,\Pi)$ procedure returns $1$ if $\sigma$ collides any path in $\Pi$; i.e., $\sigma(t)\in\mathcal{X}_{\rm free} \triangleq \bigotimes_{i\in \mathcal{V}_R}\mathcal{X}^F_i \cap \overline{\mathcal{X}}_{\rm col}$; otherwise returns $0$.



\paragraph{Feasible paths} Given the path sets of $\Sigma_i$ and $\sigma^{[-i]}$, ${\tt Feasible}_i(\Sigma_i,\sigma^{[-i]})$ is the set of paths $\sigma^{[i]}\in\Sigma_i$ such that for any $\sigma^{[i]}\in\Sigma_i$, it holds that ${\tt CollisionFreePath}(\sigma^{[i]},\sigma^{[-i]}) = 1$.

\paragraph{Weakly feasible paths} Given the path sets of $\Sigma_i$ and $\sigma^{[-i]}$, ${\tt WeakFeasible}_i(\Sigma_i,\sigma^{[-i]})$ is a subset of ${\tt Feasible}_i(\Sigma_i,\sigma^{[-i]})$ and consists of the paths $\sigma^{[i]}$ where for each path $\sigma^{[i]}$, there are a sequence of paths $\{\sigma^{[i]}_{\ell}\}$ with $\sigma^{[i]}_{\ell}\in\Sigma_i$ and a diminishing and non-negative sequence $\{\delta_{\ell}\}$ such that $(i)$ $\sigma^{[i]}_{\ell}$ converges to $\sigma^{[i]}$; $(ii)$ $\mathbb{B}(\sigma^{[i]}_{\ell}(t),\delta_{\ell})\in\mathcal{X}_i^F$; $(ii)$ $\|\sigma^{[i]}_{\ell}(t)-\sigma^{[j]}(t)\|\geq \epsilon+\delta_{\ell}$ for all $j\neq i$ for all $t$.

\paragraph{Strongly feasible paths} Given the path sets of $\Sigma_i$ and $\sigma^{[-i]})$, ${\tt StrongFeasible}_i(\Sigma_i,\sigma^{[-i]})$ is a subset of ${\tt Feasible}_i(\Sigma_i,\sigma^{[-i]})$ and consists of the paths where for each path $\sigma^{[i]}$, there are a sequence of paths $\{\sigma^{[i]}_{\ell}\}$ with $\sigma^{[i]}_{\ell}\in\Sigma_i$, a diminishing and non-negative sequence $\{\delta_{\ell}\}$ and $\delta>0$ such that $(i)$ $\sigma^{[i]}_{\ell}$ converges to $\sigma^{[i]}$;  $(ii)$ $\mathbb{B}(\sigma^{[i]}_{\ell}(t),\delta_{\ell})\in\mathcal{X}_i^F$; $(ii)$ $\|\sigma^{[i]}(t)-\sigma^{[j]}(t)\|\geq \epsilon+\delta$ for all $j\neq i$.

\section{iNash-trajectory Algorithm}\label{sec:iNash}

In this section, we propose the iNash-trajectory Algorithm to solve the open-loop game defined above. It is followed by the algorithm analysis and discussion.

\subsection{Algorithm statement}

The iNash-trajectory Algorithm leverages the RRG algorithm in~\cite{SK-EF:11}, iterative better response and model checking, and informally stated as follows. At each iteration~$k$, each robot~$i$ samples $\mathcal{X}_i$ once, and adds the new sample $x_\mathrm{rand}^{[i]}$ to its vertex set $V^{[i]}_k$. Robot~$i$ extends its previously generated graph $G^{[i]}_{k-1}$ towards the new sample $x_\mathrm{rand}^{[i]}$ via local steering, and obtains the new graph $G^{[i]}_k$. After they finish the construction of the new graphs, the active robots play a game on their product graph for \emph{one} round in a \emph{sequential} way. Robot~$i$ is active at time~$k$ if its goal set is reachable through some path of $G^{[i]}_k$. The active robot with the least index, say~$i$, first chooses a smaller-cost and feasible path on $G^{[i]}_k$ by performing the ${\tt BetterResponse}$ procedure in Algorithm~\ref{algorithm:BetterResponse}. Then robot~$i$ informs all other robots the new path. After that, the active robot with the second least index, say~$j$, performs the better response to update its path on $G^{[j]}_k$ and the new path is sent to other robots. The remaining active robots sequentially update the planned paths on their own graphs and announce the new paths to others. Once all the active robots finish the path updates, the game terminates for the current iteration~$k$. At the next iteration~$k+1$, the same steps are repeated. The iNash-trajectory Algorithm is formally stated in Algorithm~\ref{algorithm:iNash-trajectory}.

The iNash-trajectory algorithm is an anytime algorithm; i.e., assuming a solution can be found within the allotted planning time, then it is continually improved while planning time remains

\subsection{Discussion}

The ${\tt Extend}$ procedure is similar to that in the RRG algorithm~\cite{SK-EF:11} with the difference that the edges leaving from the new sample are not added. Instead, $G^{[i]}_k$ is identical to the auxiliary graph $G_n$ used in the proof of Theorem 38 in~\cite{SK-EF:11} for RRT$^*$. Notice that $G^{[i]}_k$ is a directed graph and does not include any directed circle. So there are a finite number of paths for the root to reach any leaf vertex and the ${\tt PathGeneration}$ procedure in Algorithm~\ref{algorithm:BetterResponse} is well-defined.

The tree structure returned by RRT$^*$ in~\cite{SK-EF:11} is more computationally efficient than the graph $G^{[i]}_k$ in our algorithm. However, the rewiring step in RRT$^*$ induces that $G^{[i]}_{k-1}$ may not be a subgraph of $G^{[i]}_k$. This property is crucial for the algorithm convergence to Nash equilibria in the next section.
To verify if $(\sigma^{[i]} \cap G_k^{[i]}) \in [\Phi_i]$ on Line 4 of the ${\tt BetterResponse}$ procedure, we check if the sequence $\xi_i = (\sigma^{[i]}\cap G^{[i]}_k)$ satisfies $\Phi_i$ by translating $\Phi_i$ into the corresponding Buchi automaton.


\begin{algorithm}[t] \small
  \For {$i = 1 : N$}{$V^{[i]}(0) \leftarrow x_\mathrm{init}^{[i]}$\; $E^{[i]}(0) \leftarrow \emptyset$;}
  $A_k \leftarrow \emptyset$\;
  $k\leftarrow 1$\;
  \While {$k < K$}{
  \For {$i = 1 : N$}{
    $x_\mathrm{rand}^{[i]} \leftarrow {\tt Sample}(\mathcal{X}_i)$\;
    $G^{[i]}_k\leftarrow {\tt Extend}(G^{[i]}_{k-1},x_\mathrm{rand}^{[i]})$\;}

  \For{$i\in \mathcal{V}_R\setminus A_{k-1}$}{\If{$V^{[i]}_k\cap \mathcal{X}_i^G\neq\emptyset$}{$A_k\leftarrow A_{k-1}\cup\{i\}$\;}}

  \For {$i \in A_k$}{
    $\tilde{\sigma}^{[i]}_k = \sigma^{[i]}_{k-1}$;}

  \For {$i = 1 : N$}{
  \If {$i \in A_k$}{
    $\Pi^{[i]}_k \leftarrow \{\{\sigma^{[j]}_k\}_{j\in A_k, j<i},\{\tilde{\sigma}^{[j]}_k\}_{j\in A_k, j>i}\}$\;
    $\sigma^{[i]}_k \leftarrow {\tt BetterResponse}(G^{[i]}_k, \Pi^{[i]}_k)$;}}
    $k \leftarrow k + 1$;
  }
  \caption{The iNash-trajectory Algorithm}
  \label{algorithm:iNash-trajectory}
\end{algorithm}

\begin{algorithm}[t] \small
  $V \leftarrow V^{[i]}_{k-1}$\;
  $E \leftarrow E^{[i]}_{k-1}$\;
  $x_\mathrm{nearest} \leftarrow {\tt Nearest} (E, x_\mathrm{rand}^{[i]})$\;
  $x_\mathrm{new} \leftarrow {\tt Steer} (x_\mathrm{nearest}, x_\mathrm{rand}^{[i]})$\;
  \If {${\tt ObstacleFree} (x_\mathrm{nearest},x_\mathrm{new})$} {
  $\mathcal{X}_\mathrm{near}\leftarrow {\tt NearVertices}(E,x_\mathrm{new},\min\{\gamma(\frac{\log k}{k})^{\frac{1}{n}},\eta\})$\;
  $V\leftarrow V\cup\{x_\mathrm{new}\}$\;
  \For {$x_\mathrm{near}\in \mathcal{X}_\mathrm{near}$}{ \If{${\tt ObstacleFree}(x_\mathrm{nearest},x_\mathrm{new})$}{$E\leftarrow E\cup\{(x_\mathrm{nearest},x_\mathrm{new})\}$\;}}}
  \Return{$G = (V, E)$}
  \caption{The ${\tt Extend}$ Procedure}
  \label{algorithm:extend-mRRG}
\end{algorithm}

\begin{algorithm}[t] \small
  $\mathbb{P}^{[i]}_k \leftarrow {\tt PathGeneration}(G^{[i]}_k)$\;
  $\mathbb{P}^{[i]}_{\rm f} \leftarrow \emptyset$\;
  \For{$\sigma^{[i]}\in \mathbb{P}^{[i]}_k$}{\If{${\tt CollisionFreePath}(\sigma^{[i]},\Pi^{[i]}_k) == 1\;\&\& \; (\sigma^{[i]}\cap G^{[i]}_k)\in [\Phi_i]$}{$\mathbb{P}^{[i]}_{\rm f}\leftarrow \mathbb{P}^{[i]}_{\rm f}\cup\{\sigma^{[i]}\}$\;}}
  $\sigma^{[i]}_{\rm min} \leftarrow \sigma^{[i]}_{k-1}$\;
  \For{$\sigma^{[i]} \in \mathbb{P}^{[i]}_{\rm f}$}{\If{${\tt Cost}(\sigma^{[i]}) \prec {\tt Cost}(\sigma^{[i]}_{\rm min})$}{$\sigma^{[i]}_{\rm min}\leftarrow \sigma^{[i]}$\;
  Break\;}}
  \Return{$\sigma^{[i]}_{\rm min}$}
  \caption{The ${\tt BetterResponse}$ Procedure}
  \label{algorithm:BetterResponse}
\end{algorithm}


\subsection{Analysis}

In this section, we analyze the asymptotic optimality, computational complexity, and communication cost of the iNash-trajectory Algorithm. Before doing that, we first prove the existence of Nash equilibrium.

\begin{lemma}[Existence of Nash equilibrium] It holds that $\Pi_{\rm SO} \subseteq \Pi_{\rm NE}$ and $\Pi_{\rm NE}$ is non-empty.\label{lem1}
\end{lemma}

\begin{proof} We will show $\Pi_{\rm SO} \subseteq \Pi_{\rm NE}$ by contradiction. Assume $\Pi_{\rm SO} \cap \overline{\Pi}_{\rm NE} \neq \emptyset$ with $\overline{\Pi}_{\rm NE}$ being the complement of $\Pi_{\rm NE}$, and pick any $(\sigma^{[i]})_{i\in \mathcal{V}_R}\in\Pi_{\rm SO} \cap \overline{\Pi}_{\rm NE}$. Since $(\sigma^{[i]})_{i\in \mathcal{V}_R}\notin\Pi_{\rm NE}$, there is some $i\in \mathcal{V}_R$ and $\bar{\sigma}^{[i]}\in {\tt Feasible}_i(\Sigma_i,\sigma^{[-i]})$ such that ${\tt Cost}(\bar{\sigma}^{[i]}) \prec {\tt Cost}(\sigma^{[i]})$.

Let $(\bar{\sigma}^{[i]})_{i\in \mathcal{V}_R}$ with $\bar{\sigma}^{[j]} = \sigma^{[j]}$ for $j\neq i$. Note that $\bar{\sigma}^{[j]}\in {\tt Feasible}_j(\Sigma_j,\bar{\sigma}^{[-j]})$ for all $j\neq i$, and ${\tt Cost}(\bar{\sigma}^{[i]}) = {\tt Cost}(\sigma^{[i]})$. It implies that $\bigoplus_{i\in \mathcal{V}_R}{\tt Cost}(\bar{\sigma}^{[i]}) \prec \bigoplus_{i\in \mathcal{V}_R}{\tt Cost}(\sigma^{[i]})$ contradicting that $(\sigma^{[i]})_{i\in \mathcal{V}_R}$ is socially optimal. Therefore, $\Pi_{\rm SO} \subseteq \Pi_{\rm NE}$. Since $\Pi_{\rm SO}$ is non-empty, so is $\Pi_{\rm NE}$.
\end{proof}

In what follows, we establishes the following two instrumental results for asymptotic optimality. In particular, Lemma~\ref{lem3} shows that given $\sigma^{[-i]}$, any path in ${\tt WeakFeasible}_i(\Sigma_i,\sigma^{[-i]})$ can be approximated by a sequence of paths in $\{{\tt Feasible}_i(G^{[i]}_k,\sigma^{[-i]})\}_{k\geq0}$. Lemma~\ref{lem4} shows that ${\tt Strong Feasible}_i(\Sigma_i,\bar{\sigma}^{[-i]})$ is a subset of $\sigma^{[i]}\in {\tt Strong Feasible}_i(\Sigma_i,\sigma^{[-i]}_k)$ for all but a finite $k$ where $\{\sigma^{[-i]}_k\}_{k\geq0}$ converges to $\bar{\sigma}^{[-i]}$.

\begin{lemma} Consider any $i\in \mathcal{V}_R$ and $\sigma^{[-i]}\in \bigotimes_{j\neq i}\Sigma_j$. For any $\varepsilon>0$, there is $T\geq0$ such that for any $k\geq T$ and $\sigma^{[i]}\in {\tt WeakFeasible}_i(\Sigma_i,\sigma^{[-i]})$, there is a path $\hat{\sigma}^{[i]}_k\in {\tt Feasible}_i(G^{[i]}_k,\sigma^{[-i]})$ such that $\|{\tt Cost}(\hat{\sigma}^{[i]}_k) - {\tt Cost}(\sigma^{[i]})\|\leq\varepsilon$.\label{lem3}\end{lemma}

\begin{proof} Recall that $G^{[i]}_k$ is identical to the auxiliary graph $G_n$ used in the proof of Theorem 38 for RRT$^*$. One can finish the proof by following the lines for Theorem 38 in~\cite{SK-EF:11} and using the fact of ${\tt Cost}$ being continuous.
\end{proof}

\begin{lemma} Assume that $\{\sigma^{[j]}_k\}_{k\geq0}$ converges to $\bar{\sigma}^{[j]}$ for each $j\in \mathcal{V}_R$. There is $T'\geq0$ such that for each $\sigma^{[i]}\in {\tt Strong Feasible}_i(\Sigma_i,\bar{\sigma}^{[-i]})$, it holds that $\sigma^{[i]}\in {\tt Strong Feasible}_i(\Sigma_i,\sigma^{[-i]}_k)$ for all $k\geq T'$.\label{lem4}\end{lemma}

\begin{proof} Since $\{\sigma^{[j]}_k\}_{k\geq0}$ converges to $\bar{\sigma}^{[j]}$, there is $T'\geq0$ such that $\|\bar{\sigma}^{[j]}(t)-\sigma^{[j]}_k(t)\|\leq \frac{\delta}{2}$ for all $k\geq T'$ and $j\in \mathcal{V}_R$.

Recall that $\sigma^{[i]}\in {\tt Strong Feasible}_i(\Sigma_i,\bar{\sigma}^{[-i]})$. Then there are $\{\sigma^{[i]}_{\ell}\}$ converging to $\sigma^{[i]}$, $\{\delta_{\ell}\}$ and $\delta > 0$ such that $\mathbb{B}(\sigma^{[i]}_{\ell}(t),\delta_{\ell})\in \mathcal{X}_i^F$ and $\|\sigma^{[i]}(t)-\bar{\sigma}^{[j]}(t)\|\geq \epsilon+\delta$. This further implies that $\|\sigma^{[i]}(t)-\sigma^{[j]}_k(t)\|\geq \epsilon+\frac{\delta}{2}$ for all $k\geq T'$. We then conclude that $\sigma^{[i]}\in {\tt Strong Feasible}_i(\Sigma_i,\sigma^{[-i]}_k)$ for all $k\geq T'$.
\end{proof}

Let $\hat{\Sigma}$ be the set of limit points of $\{\sigma^{[i]}_k\}_{i\in \mathcal{V}_R}$. We are ready to show the convergence of the iNash-trajectory Algorithm.

\begin{theorem}[Asymptotic optimality] There is $\bar{c}^{[i]}_{\ell}\geq0$ such that ${\tt Cost}_{\ell}(\bar{\sigma}^{[i]}) = \bar{c}^{[i]}_{\ell}$ for any $\{\bar{\sigma}^{[i]}\}_{i\in \mathcal{V}_R}\in\hat{\Sigma}$. In addition, any limit point $\{\bar{\sigma}^{[i]}\}_{i\in \mathcal{V}_R}$ is a Nash equilibrium.\label{the2}
\end{theorem}

\begin{proof} First of all, notice that $A_k \subseteq A_{k+1}$. It renders that the probability of $A_k \neq \mathcal{V}_R$ exponentially decays to zero. Without loss of any generality, we consider $A_k = \mathcal{V}_R$ for $k\geq0$ in the remainder of the proof.

Consider iteration $k$ and assume that robots $1,\cdots,i-1$ have updated their paths via Step 16 in the iNash-trajectory Algorithm. It is the turn of robot~$i$ to update the path via Step 16. Since $G^{[i]}_{k}\supseteq G^{[i]}_{k-1}$, so $\tilde{\sigma}^{[i]}_k = \sigma^{[i]}_{k-1}\in {\tt Feasible}_i(G^{[i]}_k,\Pi^{[i]}_k)$. This implies \begin{align*}{\tt Cost}(\sigma^{[i]}_k) \preceq {\tt Cost}(\tilde{\sigma}^{[i]}_k) = {\tt Cost}(\sigma^{[i]}_{k-1}).\end{align*} Hence, for each $\ell\in\{1,\cdots,p\}$, the sequence $\{{\tt Cost}_{\ell}(\sigma^{[i]}_k)\}_{k\geq0}$ is non-increasing. Since ${\tt Cost}_{\ell}(\sigma^{[i]}_k)$ is lower bounded by zero, so the sequence $\{{\tt Cost}_{\ell}(\sigma^{[i]}_k)\}$ converges to some constant $\bar{c}^{[i]}_{\ell}\geq0$. Hence, we have ${\tt Cost}_{\ell}(\bar{\sigma}^{[i]}) = \bar{c}^{[i]}_{\ell}$ for any limit point $\{\bar{\sigma}^{[i]}\}_{i\in \mathcal{V}_R}\in\hat{\Sigma}$.

We proceed to show that any limit point $\{\bar{\sigma}^{[i]}\}_{i\in \mathcal{V}_R}\in\hat{\Sigma}$ is a Nash equilibrium. There is a subsequence $\pi_i\triangleq\{k_i\}$ such that $\{\sigma_k^{[i]}\}_{k\in \pi_i}$ converge to $\bar{\sigma}^{[i]}$. Let $\pi \triangleq \cap_{i\in \mathcal{V}_R}\pi_i$. Without loss of any generality, we assume that $\pi = \mathbb{Z}_{\geq0}$ in the remainder of the proof.

For each $i\in \mathcal{V}_R$, there is no $\tilde{\sigma}^{[i]}_k\in {\tt Feasible}_i(G^{[i]}_k,\Pi^{[i]}_k)$ such that \begin{align}{\tt Cost}(\sigma^{[i]}_k) \succeq {\tt Cost}(\tilde{\sigma}^{[i]}_k).\label{e5}\end{align}

Pick any $\varepsilon>0$. By Lemma~\ref{lem3} with $\sigma^{[-i]} = \Pi^{[i]}_k$, there is $T$ such that for any $\hat{\sigma}^{[i]}_k\in {\tt WeakFeasible}_i(\Sigma_i,\Pi^{[i]}_k)$, there is $\tilde{\sigma}^{[i]}_k\in {\tt Feasible}_i(G^{[i]}_k,\Pi^{[i]}_k)$ such that \begin{align}{\tt Cost}(\tilde{\sigma}^{[i]}_k) \preceq {\tt Cost}(\hat{\sigma}^{[i]}_k)+\varepsilon \textbf{1}_p,\label{e6}\end{align} where $\textbf{1}_p$ is the column vector with $p$ ones.

The combination of~\eqref{e5} and~\eqref{e6} renders that there is no $\hat{\sigma}^{[i]}_k\in {\tt WeakFeasible}_i(\Sigma_i,\Pi^{[i]}_k)$ and $k\geq T$ such that: \begin{align}{\tt Cost}(\sigma^{[i]}_k)\succeq{\tt Cost}(\hat{\sigma}^{[i]}_k)+\varepsilon\textbf{1}_p.\label{e7}\end{align}

Notice that ${\tt StrongFeasible}_i(\Sigma_i,\Pi^{[i]}_k)\subseteq{\tt WeakFeasible}_i(\Sigma_i,\Pi^{[i]}_k)\subseteq {\tt Feasible}_i(\Sigma_i,\Pi^{[i]}_k)$. Combine this property,~\eqref{e6} and Lemma~\ref{lem4} with $\sigma^{[i]}_k = \Pi^{[i]}_k$, and we reach that there is $T'\geq T$ such that there is no $\sigma^{[i]}\in {\tt StrongFeasible}_i(\Sigma_i,\bar{\sigma}^{[-i]})$ and $k\geq T'$ such that: \begin{align}{\tt Cost}(\sigma^{[i]}_k)\succeq {\tt Cost}(\sigma^{[i]})+\varepsilon\textbf{1}_p.\label{e2}\end{align}

Take the limit on $k$ in~\eqref{e2}. Then there is no ${\sigma}^{[i]}\in {\tt StrongFeasible}_i(\Sigma_i,\bar{\sigma}^{[-i]})$ such that \begin{align}{\tt Cost}(\bar{\sigma}^{[i]})\succeq {\tt Cost}({\sigma}^{[i]})+\varepsilon\textbf{1}_p.\label{e3}\end{align} Since~\eqref{e3} holds for $\varepsilon>0$, no ${\sigma}^{[i]}\in {\tt StrongFeasible}_i(\Sigma_i,\bar{\sigma}^{[-i]})$ exists such that \begin{align}{\tt Cost}(\bar{\sigma}^{[i]})\succeq {\tt Cost}({\sigma}^{[i]}).\label{e4}\end{align} Since~\eqref{e4} holds for any~$i\in \mathcal{V}_R$, it establishes that the limit point $\{\bar{\sigma}^{[i]}\}_{i\in \mathcal{V}_R}$ is a Nash equilibrium.

The completeness of the algorithm is a direct result of Lemma~\ref{lem1} and the asymptotic convergence to $\Pi_{\rm NE}$. It completes the proof.
\end{proof}

Next, we will analyze the computational complexity of the algorithm in terms of the ${\tt CollisionFreePath}$ procedure. Let $\theta_n$ to be the total number of calls to the ${\tt CollisionFreePath}$ procedure at iteration~$n$.

\begin{lemma}[Computational complexity] It holds that $\theta_n \leq \bigoplus_{i\in \mathcal{V}_R}|\mathbb{P}^{[i]}_k|$, where $\mathbb{P}^{[i]}_k$ is defined in Algorithm~\ref{algorithm:BetterResponse}.\label{lem2}
\end{lemma}

In Lemma~\ref{lem2}, the quantity $|\mathbb{P}^{[i]}_k|$ is independent of the robot number. So the worst computational complexity of the iNash-trajectory grows linearly in the robot number. It is in contrast to the exponential dependency in centralized path planning. The computational efficiency comes with the non-cooperative game theoretic formulation where each robot myopically responds to others. Note that a Nash equilibrium may not be socially optimal for the robot team.

Let $\vartheta_n$ to be the number of exchanged paths in iteration~$n$. The following lemma shows the worst communication cost is linear in the robot number.

\begin{lemma}[Communication cost] $\vartheta_n \leq 2N$.\label{lem5}
\end{lemma}

\begin{proof} At iteration~$n$, all the active robots broadcast their planned trajectories before the path updates. Then the active robots sequentially update and broadcast the new paths. So $\vartheta_n = 2|A_k| \leq 2N$.
\end{proof}

\subsection{Comparison}


In order to demonstrate the scalability of iNash-trajectory Algorithm, we consider the benchmark algorithm, the iOptimalControl Algorithm. The key difference between the iOptimalControl and iNash-trajectory Algorithms is that a centralized authority in the iOptimalControl Algorithm determines a social optimum on the product graph at each iteration. In the iOptimalControl Algorithm, we use the notation $(\sigma\cap \bigotimes_{i\in A_k}G^{[i]}_k)\in[\Phi]$ for $(\sigma^{[i]}\cap G^{[i]}_k)\in[\Phi_i]$ for all $i\in A_k$.


\begin{algorithm}[t] \small
  \For {$i = 1 : N$}{$V^{[i]}(0) \leftarrow x_\mathrm{init}^{[i]}$\; $E^{[i]}(0) \leftarrow \emptyset$;}
  $A_k \leftarrow \emptyset$\;
  $k\leftarrow 1$\;
  \While {$k < K$}{
  \For {$i = 1 : N$}{
    $x_\mathrm{rand}^{[i]} \leftarrow {\tt Sample}(\mathcal{X}_i)$\;
    $G^{[i]}_k\leftarrow {\tt Extend}(G^{[i]}_{k-1},x_\mathrm{rand}^{[i]})$\;}

  \For{$i\in \mathcal{V}_R\setminus A_{k-1}$}{\If{$V^{[i]}_k\cap \mathcal{X}_i^G\neq\emptyset$}{$A_k\leftarrow A_{k-1}\cup\{i\}$\;}}

  $(\sigma^{[i]}_k)_{i\in A_k}\leftarrow{\tt OptimalTrajectory}(\bigotimes_{i\in A_k}G^{[i]}_k)$\;

  $k \leftarrow k + 1$;
  }
  \caption{The iOptimalControl Algorithm}
  \label{algorithm:iOptimalControl}
\end{algorithm}

\begin{algorithm}[t] \small
  \For{$i\in A_k$}{$\mathbb{Q}^{[i]}_k \leftarrow {\tt PathGeneration}(G^{[i]}_k)$\;}

  \For{$i\in A_k$}{
  $\mathbb{P}^{[i]}_{\rm f} \leftarrow \emptyset$\;
  \For{$\sigma^{[i]}\in \mathbb{Q}^{[i]}_k$}{\If{${\tt CollisionFreePath}(\sigma^{[i]},\mathbb{Q}^{[-i]}_k) == 1\;\&\&\; (\sigma\cap \bigotimes_{i\in A_k}G^{[i]}_k)\in[\Phi]$}{$\mathbb{P}^{[i]}_{\rm f}\leftarrow \mathbb{P}^{[i]}_{\rm f}\cup\{\sigma^{[i]}\}$\;}}}

  $\sigma_{\rm min} \leftarrow {\tt Sample}(\bigotimes_{i\in A_k}\mathbb{P}^{[i]}_{\rm f})$\;

  \For{$\sigma \in \bigotimes_{i\in A_k}\mathbb{P}^{[i]}_{\rm f}$}{\If{$\bigoplus_{i\in A_k}{\tt Cost}(\sigma^{[i]}) \prec \bigoplus_{i\in A_k}{\tt Cost}(\sigma^{[i]}_{\rm min})$}{$\sigma_{\rm min}\leftarrow \sigma$\;}}
  \Return{$\sigma_{\rm min}$}
  \caption{The ${\tt OptimalTrajectory}$ Procedure}
  \label{algorithm:OptimalTrajectory}
\end{algorithm}

The following theorem guarantees the asymptotic optimality of the iOptimalControl Algorithm.

\begin{theorem}[Asymptotic optimality] Any limit point $\{\bar{\sigma}^{[i]}\}_{i\in \mathcal{V}_R}$ is a social optimum.\label{the3}
\end{theorem}

\begin{proof} It is a direct result of Theorem 38 in~\cite{SK-EF:11}.
\end{proof}

Next, we will analyze the computational complexity of the algorithm in terms of the ${\tt CollisionFreePath}$ procedure. Let $\theta_n'$ to be the total number of calls to the ${\tt CollisionFreePath}$ procedure at iteration~$n$.

\begin{lemma}[Computational complexity] It holds that $\theta'_n = \bigotimes_{i\in \mathcal{V}_R}|\mathbb{Q}^{[i]}_k|$, where $\mathbb{Q}^{[i]}_k$ is defined in Algorithm~\ref{algorithm:OptimalTrajectory}.\label{lem6}
\end{lemma}

The above lemma shows that the computational complexity exponentially grows vs.\ robot number. Table~\ref{ta:com2} summarizes the comparison of the iOptimalControl and iNash-trajectory Algorithms where the prices of anarchy and stability are compared for the case $p = 1$. In particular, the price of stability (POS)~\cite{Nisan.Roughgarden.ea:07} is the ratio between the minimum additive cost function value in $\Pi_{\rm NE}$ and that of one in $\Pi_{\rm SO}$, and defined as follows:
\begin{align*}{\rm POS} = \frac{\inf_{\sigma\in\Pi_{\rm NE}}\bigoplus_{i\in \mathcal{V}_R}{\tt Cost}(\sigma^{[i]})}{\bigoplus_{i\in \mathcal{V}_R}{\tt Cost}(\bar{\sigma}^{[i]})},\end{align*} for any $\bar{\sigma}\in\Pi_{\rm SO}$. By Lemma~\ref{lem1}, we know $\Pi_{\rm SO}\subseteq \Pi_{\rm NE}$ and thus POS is equal to 1. On the other hand, the price of anarchy (POA)~\cite{Nisan.Roughgarden.ea:07} is the ratio between the maximum additive cost function value in $\Pi_{\rm NE}$ and that of one in $\Pi_{\rm SO}$, and given by: \begin{align*}{\rm POA} = \frac{\sup_{\sigma\in\Pi_{\rm NE}}\bigoplus_{i\in \mathcal{V}_R}{\tt Cost}(\sigma^{[i]})}{\bigoplus_{i\in \mathcal{V}_R}{\tt Cost}(\bar{\sigma}^{[i]})},\end{align*} for any $\bar{\sigma}\in\Pi_{\rm SO}$. The value of POA depends on a number of factors; e.g., the obstacle locations, the dynamic systems and so on. It is interesting to find a lower bound on $\Pi_{\rm SO}$ given more information as in; e.g.,~\cite{Johari.Tsitsiklis:04}, and utilize mechanism design to eliminate the price of anarchy.

\begin{table}[ht]
\caption{The comparison of the iOptimalControl and iNash-trajectory Algorithms}
\centering
\begin{tabular}{|c|c|c|}
  \hline
  & \text{iOptimalControl} & \text{iNash-trajectory}
\tabularnewline
  \hline
  \text{Solution Notion}  & \text{Social optimum} & \text{Nash equilibrium}
\tabularnewline
  \hline
  \text{Solution Feasibility}  & \text{Yes} & \text{Yes}
\tabularnewline
  \hline
  \text{Price of stability}  & \text{N/A} & \text{One}
\tabularnewline
  \hline
  \text{Price of anarchy}  & \text{N/A} & \text{Unknown}
\tabularnewline
  \hline
  \text{Coordination}  & \text{High} & \text{Low}
\tabularnewline
  \hline
  \text{Asymptotic optimality}  & \text{Yes} & \text{Yes}
\tabularnewline
  \hline
  \text{Computational complexity}  & \text{Exponential} & \text{Linear}
\tabularnewline
  \hline
\end{tabular}\label{ta:com2}
\end{table}

\section{Experiments}\label{sec:experiments}

We perform two experiments in simulation to evaluate the performance of iNash. The first involves 8 circular robots moving in an environment with randomly generated obstacles (Figure~\ref{fig:exp1}-left), while the second involves 6 robots in a traffic intersection scenario (Figure~\ref{fig:exp1}-right); both involve state spaces consisting of first order dynamics and time. Robots are holonomic discs with radii of $0.5$ meters. 

We compare iNash to two prioritized methods that are not guaranteed to return a Nash-Equilibrium, but are arguably similar to our proposed algorithm. The first is the standard prioritized approach from \cite{Erdmann.Lozano:87}. Each robot builds its own random graph; then robots select their paths in order such that the path of robot $i$ does not conflict with robots ${1, \ldots, i-1}$.
The second is an any-time version of the prioritized method. Each time robot $i$ finds a better path that does not conflict with the paths of robots $1, \ldots, i-1$, then for $j = i+1, i+2, \ldots$ (in order) robot $j$ must choose a new path that does not conflict with robots $1, \ldots, j-1$. This differs from our algorithm (where new paths must respect the paths of all other robots), and the solution is not guaranteed to converge to a Nash Equilibrium.

For experiments we consider the task specifications for each robot to be of the form
$\Phi_i = \mathbf{F}\ p_{G_i} \wedge \mathbf{G}\ p_F$, i.e., each robot tries to reach a different goal region in the shortest possible distance
while respecting the same constraint set $\mathcal{X}^F$. The automaton consists of two states (see Fig.~\ref{fig:buchi}).

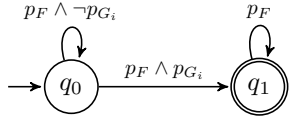
\begin{figure}[!t]
\centering
\usetikzlibrary{arrows,positioning,automata,shadows,fit,shapes}
\begin{tikzpicture}[->,>=stealth',shorten >=1pt,auto,node distance=1.5cm, semithick, initial text=]
\tikzset{every state/.style={minimum size=1pt}}

\node[state, initial] 			at (0,0)		(q0) {$q_0$};
\node[state, accepting]      	at (2.5, 0)   	(q1) {$q_1$};

\begin{scope}[every node/.style={scale=.8}]
	\path 	(q0) edge [loop above] node {$p_F \wedge \neg p_{G_i}$} (q0)
				edge node {$p_F \wedge p_{G_i}$} (q1)
			(q1) edge [loop above] node {$p_F$} (q1);
\end{scope}
\end{tikzpicture}
\caption{\small Automaton for $\Phi_i = \mathbf{F}\ p_{G_i} \wedge \mathbf{G}\ p_F$.}
\label{fig:buchi}
\end{figure}

\begin{figure*}[t]
    \centering
 \includegraphics[width=0.3\linewidth, trim=155 230 155 230, clip=true]{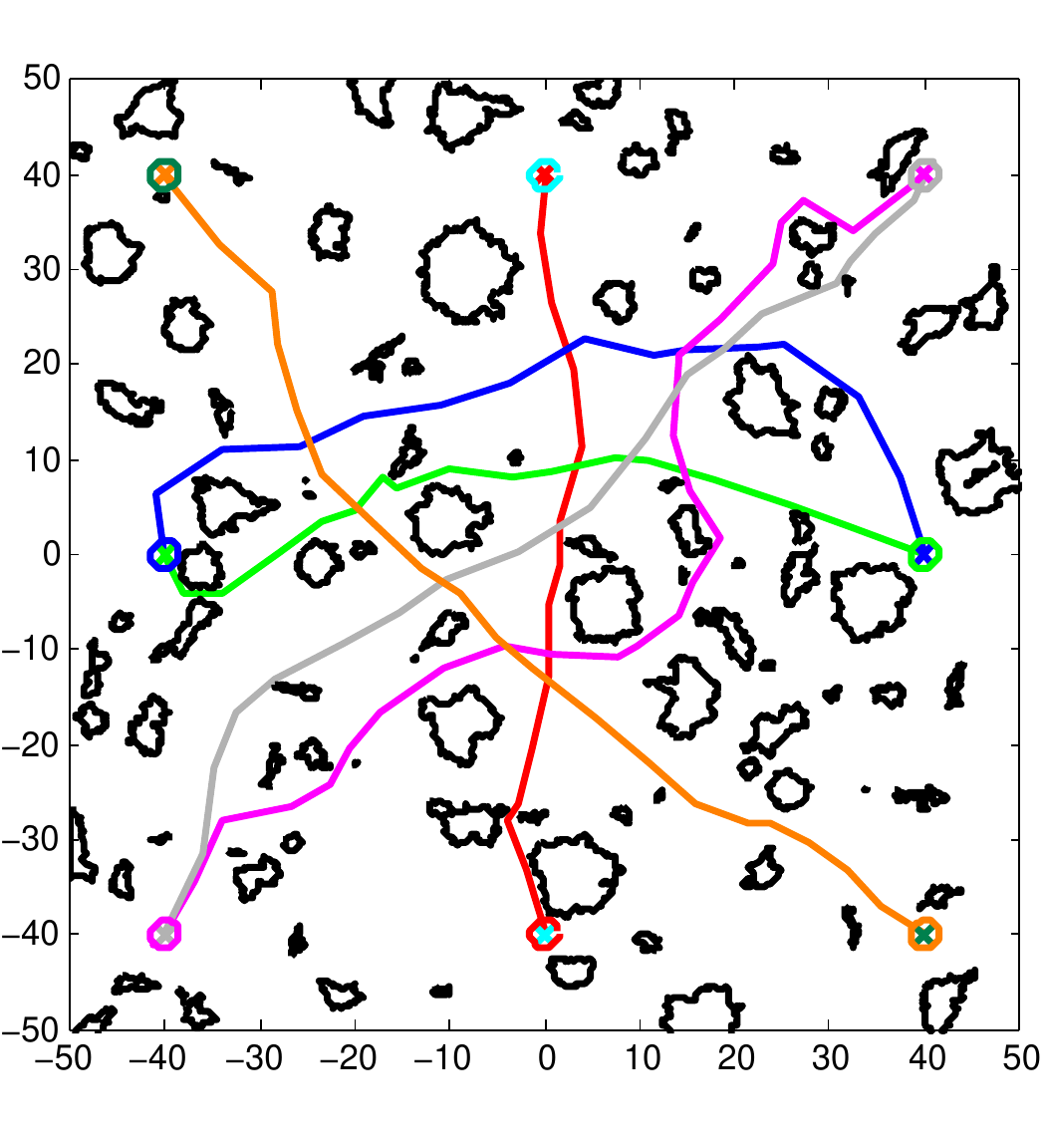}
 \hspace{0.2in}
 \includegraphics[width=0.3\linewidth, trim=155 230 155 230, clip=true]{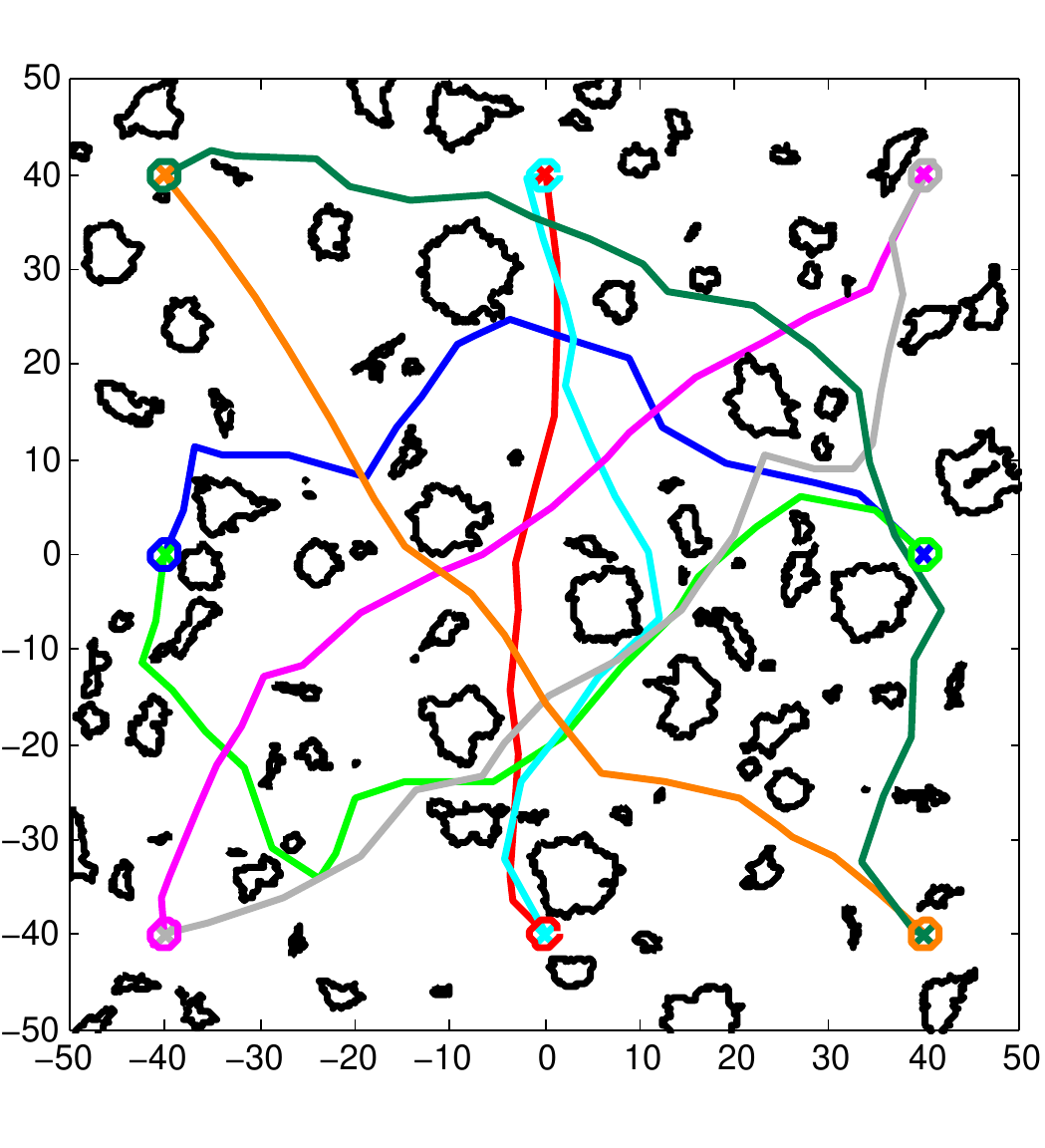}
 \hspace{0.2in}
 \includegraphics[width=0.3\linewidth, trim=155 230 155 230, clip=true]{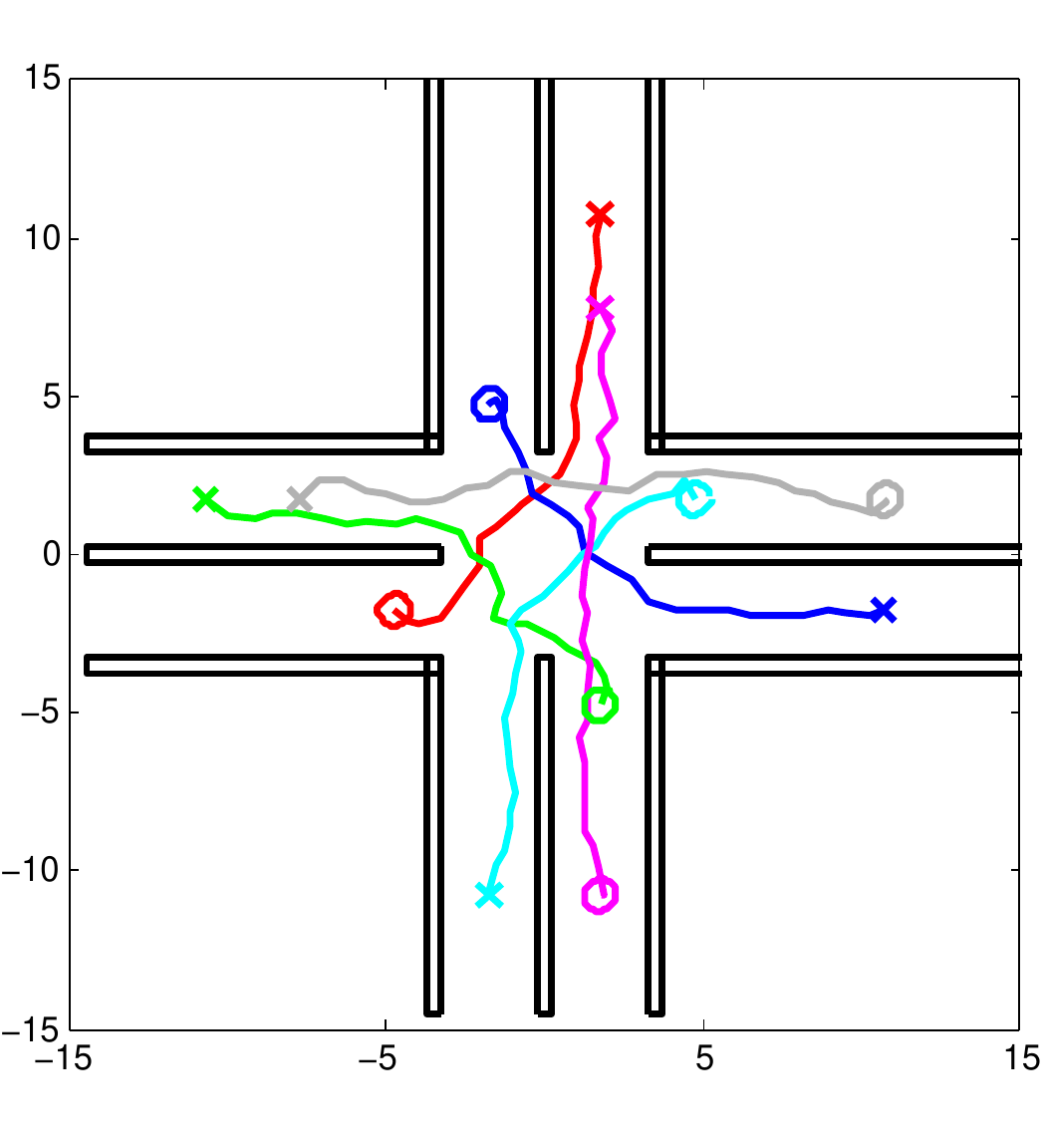}

   \caption[]{\small Experimental environments. Obstacles are outlined in black, paths are colored lines. Robots start at `O's and end at `X's ('O's are drawn 3x the robot radii to help visualization). {\bf Left/Center:} 8 robots in a randomly generated environment; the Nash Equilibrium in the left trial allows 6 of 8 robots to reach their goals, while all 8 reach their goals in the center trial. {\bf Right:} 6 Robots at an intersection and the paths corresponding to a Nash equilibrium where all robots reach their goals}

\label{fig:exp1}
\end{figure*}

%
%
%
%

%
%
%

\newlength{\mycolw}
\setlength{\mycolw}{.69cm}

\begin{table*}[!t]

  \caption[Results]{Experimental Results, Random Environment}
  \label{table:A}

\begin{minipage}[b]{0.67\linewidth}

  \renewcommand{\arraystretch}{1.2}
  \begin{tabular}{ | l | x{\mycolw} | x{\mycolw} | x{\mycolw} | x{\mycolw} | x{\mycolw} | x{\mycolw} | x{\mycolw} | x{\mycolw} |}
     \multicolumn{9}{l}{\hspace{2.7cm} Mean path length over 20 trials (ratio vs.\ socially optimal length)}\tabularnewline
  \hline
     \multirow{2}{*}{Algorithm} & \multicolumn{8}{|c|}{ Robot ID}\tabularnewline
  \cline{2-9}
      & 1 & 2 & 3 & 4 & 5 & 6 & 7 & 8\tabularnewline
  \cline{1-9}
      iNash (any-time)      & $1.295$ & $1.343$ & $1.324$ & $1.301$ & $1.166$ & $1.293$ & $1.224$ & $1.202$\tabularnewline
  \cline{1-9}
      Prioritized & $1.084$ & $1.149$ & $1.263$ & $1.316$ & $1.228$ & $1.326$ & $1.312$ & $1.349$\tabularnewline
  \cline{1-9}
      Prioritized (any-time)  & $1.081$ & $1.129$ & $1.153$ & $1.200$ & $1.126$ & $1.168$ & $1.163$ & $1.204$\tabularnewline
  \hline
  \end{tabular}

\end{minipage}
\begin{minipage}[b]{0.3\linewidth}

  \renewcommand{\arraystretch}{1.2}
  \begin{tabular}{ | c | c | c | c | c | c | c | c |}
     \multicolumn{8}{l}{Total times reached goal (of 20)}\\
  \hline
      \multicolumn{8}{|c|}{ Robot ID}\\
  \cline{1-8}
       1 & 2 & 3 & 4 & 5 & 6 & 7 & 8\\
  \cline{1-8}
      20 & 19 & 18 & 20 & 16 & 15 & 18 & 17\\
  \cline{1-8}
      20 & 20 & 18 & 19 & 20 & 18 & 18 & 15\\
  \cline{1-8}
      20 & 20 & 18 & 20 & 19 & 20 & 20 & 20\\
  \hline
  \end{tabular}

\end{minipage}

\end{table*}

\newlength{\mycolwd}
\setlength{\mycolwd}{.9cm}

\begin{table*}[!t]

  \caption[Results]{Experimental Results, Intersection Environment}
  \label{table:B}

\begin{minipage}[b]{0.67\linewidth}

  \renewcommand{\arraystretch}{1.2}
  \begin{tabular}{ | l | x{\mycolwd} | x{\mycolwd} | x{\mycolwd} | x{\mycolwd} | x{\mycolwd} | x{\mycolwd} |}
     \multicolumn{7}{l}{\hspace{2.7cm} Mean path length over 20 trials (ratio vs.\ socially optimal length)}\tabularnewline
  \hline
     \multirow{2}{*}{Algorithm} & \multicolumn{6}{|c|}{ Robot ID}\tabularnewline
  \cline{2-7}
      & 1 & 2 & 3 & 4 & 5 & 6\tabularnewline
  \cline{1-7}
      iNash (any-time) & $1.168$ & $1.228$ & $1.245$ & $1.243$ & $1.201$ & $1.166$\tabularnewline
  \cline{1-7}
      Prioritized & $1.107$ & $1.238$ & $1.384$ & $1.492$ & $1.339$ & $1.300$\tabularnewline
  \cline{1-7}
      Prioritized (any-time)  & $1.085$ & $1.212$ & $1.165$ & $1.248$ & $1.136$ & $1.251$\tabularnewline
  \hline
  \end{tabular}

\end{minipage}
\begin{minipage}[b]{0.3\linewidth}

  \renewcommand{\arraystretch}{1.2}
  \begin{tabular}{  | c | c | c | c | c | c |}
     \multicolumn{6}{l}{Total times reached goal (of 20)}\\
  \hline
     \multicolumn{6}{|c|}{ Robot ID}\\
  \cline{1-6}
       1 & 2 & 3 & 4 & 5 & 6\\
  \cline{1-6}
       11 & 14 & 12 & 14 & 16 & 15\\
  \cline{1-6}
       19 & 18 & 19 & 13 & 12 & 4\\
  \cline{1-6}
       19 & 18 & 20 & 19 & 20 & 15\\
  \hline
  \end{tabular}

\end{minipage}

\end{table*}

%
%
%
%
%
%

\emph{Discussion of Experimental Results:} Experimental results are summarized in tables \ref{table:A}-\ref{table:B}. In iNash all robots tend to bear the burden of conflict resolution similarly, on average. This contrasts with the prioritized methods, in which robot's with lower IDs have shorter paths and reach the goal more frequently than robots with higher IDs.
The result that some iNash paths are longer than the prioritized paths is expected, given that in iNash robots act in their own self interest. 

\section{Conclusions}

This paper discusses a class of multi-robot motion planning problems where each robot is associated with multiple-objectives and independent class specifications. We formulated the problem as an open-loop, non-cooperative differential game and proposed a distributed, anytime algorithm to compute the Nash equilibrium. Techniques from Rapidly-exploring Random Graphs and iterative better response were used to provide convergence guarantees and analyse the price of stability as well as the communication cost of the algorithm. We also presented results of simulation experiments that demonstrate the efficiency and anytime nature of the algorithm. Future directions include coupled task specifications of robots and algorithms which can eliminate the price of anarchy.

\bibliographystyle{plain}

\end{document}